\documentclass[journal,10pt,twocolumn]{IEEEtran}
\usepackage[numbers,sort,square,compress]{natbib}
\usepackage[utf8]{inputenc}
\usepackage[T1]{fontenc}
\usepackage{stackengine}
\usepackage{graphicx}
\usepackage{grffile}
\usepackage{longtable}
\usepackage{wrapfig}
\usepackage{rotating}
\usepackage[normalem]{ulem}
\usepackage{amsmath}
\usepackage{amsthm}
\usepackage{textcomp}
\usepackage{amssymb}
\usepackage{capt-of}
\usepackage{hyperref}
\usepackage{color}
\usepackage{parskip}
\usepackage{amsmath}
\usepackage[bottom]{footmisc}
\usepackage{braket}
\usepackage{placeins}
\usepackage[table,xcdraw]{xcolor}
\usepackage{multicol}
\usepackage{lipsum}
\usepackage{graphicx}
\usepackage[numbers,sort,square,compress]{natbib}
\usepackage[amssymb,binary]{SIunits}
\usepackage{float}

\newtheorem{theorem}{Theorem}
\usepackage{amsmath}
\usepackage{tikz}
\usepackage{tabularx}
\usepackage{ellipsis}
\usetikzlibrary{calc}
\usepackage{tkz-euclide,subfigure}
\usetikzlibrary{decorations.pathreplacing,decorations.markings,shapes.geometric}

\renewcommand\thesection{\arabic{section}}
\renewcommand\thesubsection{\thesection.\arabic{subsection}}
\renewcommand\thesubsubsection{\thesubsection.\arabic{subsubsection}}

\renewcommand\thesectiondis{\arabic{section}}
\renewcommand\thesubsectiondis{\thesectiondis.\arabic{subsection}}
\renewcommand\thesubsubsectiondis{\thesubsectiondis.\arabic{subsubsection}}

\def\bU{\mathbf{U}}

\def\bF{\mathbf{F}}

\def\bf{\mathbf{f}}

\def\bv{\mathbf{v}}

\def\bx{\mathbf{x}}
\def\by{\mathbf{y}}

\def\bA{\mathbf{A}}
\def\bB{\mathbf{B}}

\def\bD{\mathbf{D}}

\def\bF{\mathbf{F}}
\def\bG{\mathbf{G}}
\def\bH{\mathbf{H}}
\def\bI{\mathbf{I}}
\def\bJ{\mathbf{J}}

\def\bN{\mathbf{N}}

\def\bP{\mathbf{P}}

\def\bT{\mathbf{T}}
\def\bU{\mathbf{U}}
\def\bV{\mathbf{V}}

\def\bX{\mathbf{X}}
\def\bY{\mathbf{Y}}

\newtheorem{prop}{Proposition}
\usepackage{mathrsfs}
\usepackage{etoolbox}
\makeatletter
\patchcmd{\@makecaption}
  {\scshape}
  {}
  {}
  {}
\makeatletter
\patchcmd{\@makecaption}
  {\\}
  {.\ }
  {}
  {}
\makeatother
\def\tablename{Table}

\author{
\IEEEauthorblockN{Agulla Surya Bharath, Devanshu Singh Gaharwar, Kumar Appaiah and Debasattam Pal}
}

\title{Design of Discrete-time Matrix All-Pass Filters Using Subspace
  Nevanlinna Pick Interpolation}
\begin{document}

\maketitle
\begin{abstract}
  Unitary matrix-valued functions of frequency are matrix all-pass
  systems, since they preserve the norm of the input vector
  signals. Typically, such systems are represented and analyzed using
  their unitary-matrix valued frequency domain characteristics,
  although obtaining rational realizations for matrix all-pass systems
  enables compact representations and efficient
  implementations. However, an approach to obtain matrix all-pass
  filters that satisfy phase constraints at certain frequencies was
  hitherto unknown. In this paper, we present an interpolation
  strategy to obtain a rational matrix-valued transfer function from
  frequency domain constraints for discrete-time matrix all-pass
  systems. Using an extension of the Subspace Nevanlinna Pick
  Interpolation Problem (SNIP), we design a construction for
  discrete-time matrix all-pass systems that satisfy the desired phase
  characteristics. An innovation that enables this is the extension of
  the SNIP to the boundary case to obtain efficient time-domain
  implementations of matrix all-pass filters as matrix linear constant
  coefficient difference equations, facilitated by a rational
  (realizable) matrix transfer function. We also show that the
  derivative of matrix phase constraints, related to the group delay
  at the interpolating points, can be optimized to control the
  all-pass transfer matrices at the unspecified
  frequencies. Simulations show that the proposed technique for
  unitary matrix filter design performs as well as traditional DFT
  based interpolation approaches, including Geodesic interpolation and
  the popular Givens rotation based matrix parameterization.
\end{abstract}

\section{Introduction}
\label{sec:introduction}
Filtering signals is among the most fundamental operations in signal
processing. In general, filtering scalar signals is well understood,
and there is mature theory that discusses filter design and
implementation for both analog and digital scalar filters. However,
with the increased interest in multiple-input multiple-output systems
in several allied areas, the concept of filtering vector signals has
gained importance. Designing precise filters for vector signals
(wherein the signal at each time instant is a real or complex vector)
under various constraints is also interesting from the point of view
of several practical applications, although there has not been much
work in the past in this direction. In this paper, we focus on the
design of discrete-time ``matrix'' all-pass filters, that transform a
vector signal's phase while ensuring that their norm is not altered
for all frequencies. In particular, unlike the standard practice of
using frequency domain transform techniques for filtering, we present
an interpolation based filter design technique that produces matrix
all-pass filters
for practical realizability. This idea has several applications, such
as combined left and right audio signals in case of stereo audio as
well as for feedback in control and communication systems etc. As an
example to show the effectiveness of the proposed techniques, we
consider the MIMO precoders for wireless communication systems which
employ orthogonal frequency division multiplexing (OFDM). These
precoders can be accurately and efficiently realized using time domain
techniques, as opposed to the traditionally used
approaches~\cite{precodingMIMO,LinearPrecodMIMO}.

Matrix filtering with a norm preservation constraint is typically
accomplished using frequency domain
techniques~\cite{LimitFeedbackPrecode,feedback}. Specifically, this
involves computing the Fourier transform of the signal, performing the
all-pass filtering on a per-frequency basis, and using the inverse
Fourier transform, as is common in the case of vector communication
systems~\cite{TseVis,feedback}. However, when the matrix all-pass
filter lends itself to a time domain realization, this method is not
ideal. In particular,
when the matrix
all-pass filter has an efficient linear constant coefficient
difference equation (LCCDE) realization, the filter realized using
frequency domain techniques will be inaccurate, and will also result
in less efficient realizations. To address this, we present an
interpolation based matrix all-pass filter design technique that
results in a realizable filter (that can be implemented in the time
domain as an LCCDE) while satisfying the frequency domain
constraints. Our approach extends the classical Subspace Nevanlinna
Pick Interpolation (SNIP) method~\cite{SNIP} that is well-known in the
context of control systems to the ``boundary'' case to obtain matrix
filters that satisfy some prior constraints, while ensuring that the
Fourier transform of its system function is a unitary matrix at all
frequencies, thus obtaining norm-preserving (matrix all-pass) filters.

The classical Nevanlinna interpolation problem has its roots in the problem of synthesis of dynamical systems as passive electrical networks (see~\cite{interp_posit_fns}).
This, as well as all the subsequent extensions of it, however, considers only the situation wherein the interpolating frequencies and the prescribed values of the desired transfer function are strictly within the respective critical regions. For example, in the scalar version of the SNIP dealt with in~\cite{SNIP}, the polar plot of the transfer function must lie strictly within the unit disk, and the frequencies that are given lie on the open right-half of the complex plane. It is important to note that the solution of the classical SNIP crucially depends on these strictness assumptions. In this paper, we push the SNIP to its boundary: we deal with the case wherein the desired transfer function’s polar plot is on the unit disk (i.e., all-pass), and the frequencies, too, are given on the boundary (the unit circle because we consider discrete time systems).

Our key contributions in this paper are as follows:
\begin{itemize}
\item We present an approach to realize a
  discrete-time matrix all-pass filter, when given a feasible set of
  frequency responses (unitary matrices) and group delay matrices
  for a finite set of frequencies \{$\omega_i$\}. Specifically, our solution
  yields a rational matrix $z$-transform for the
  required all-pass filter that satisfies all the given frequency domain
  conditions, and its transfer function matrix is unitary valued for
  all $\omega \in (-\pi, \pi]$. This can be
  viewed as a generalization of the Blaschke interpolation based
  approach that is specific to \emph{scalar} all-pass filter
  design~\cite{scalarallpass,bolotnikov2018boundary} to the \emph{matrix} case.

\item We obtain this \emph{matrix} all-pass filter by extending the SNIP
  technique to the boundary case. Specifically, since the Pick matrix
  in the case of standard SNIP~\cite{SNIP} becomes ill-defined when we  demand a unitary valued solution, we provide a modified approach   using the modified Pick (Schwarz-Pick)
  matrix to generalize the SNIP filter realization to the
  boundary case in discrete-time setting.

\item Finally, we also present an optimization based approach that
  tunes the slopes of the matrix phase response at specific
  frequencies to obtain realizable filters with desirable
  characteristics.
\end{itemize}
The proposed approach for filtering is both novel as well
as efficient in terms of implementation. In particular, prior
approaches to perform all-pass matrix filtering in the frequency
domain have used DFT based techniques that involve at least
$N_{\text{FFT}}$
multiplications~\cite{lou2013comparison,ieee80211}. In addition, these
techniques have largely relied on frequency domain interpolation of
precoders interpolation on manifolds~\cite{pitaval2013coding,Flagdist}
or interpolation of parameterized unitary
matrices~\cite{4114278,Givens_rot}, a technique that is employed in
recent wireless OFDM based standards as
well~\cite{ieee80211}. However, as we show in this paper, in
situations where the all-pass filter has an impulse response that can
be characterized using fewer coefficients, significant savings in
terms of computations can be realized using the SNIP based approach,
while faithfully capturing the frequency domain precoder characteristics.

The rest of the paper is organized as follows: Section~\ref{sec:motivation} outlines the importance of matrix all-pass filter design problem statement, Section~\ref{sec:SNIP} briefly describes the classical SNIP, Section~\ref{sec:DTallpass} describes the discrete-time matrix all-pass filter design problem and its solution,
Section~\ref{sec:simulations}
contains the simulations results and interpretations for some practical purposes,
Section~\ref{sec:conclusion} provides some concluding remarks and
discusses future directions.

\emph{Notation}: Unless otherwise specified, bold capital symbols refer to matrices, bold
smallcase symbols correspond to vectors, $\mathbf{I}_m$ refers to an
$m\times m$ identity matrix and $\mathbf{0}_{m\times m}$ refers to an
$m\times m$ all-zero matrix. We also use the abbreviations CT for
continuous-time and DT for discrete-time.  For any matrix $\bA$, $\bA^*$ is the conjugate transpose of $\bA$.

\section{Motivation}
\label{sec:motivation}
To motivate this problem, we first pose the \emph{scalar} all-pass filter
problem: if the frequency response of a discrete-time all-pass filter
is given for certain (finitely many) frequencies, how can we obtain an
all-pass filter that satisfies these constraints? In general, there
exist an infinite number of filters that satisfy these
constraints. Recent work has shown that, if the group delays are also
known at the given frequencies, then a realizable all-pass
filter can be obtained as a Blaschke
product~\cite{scalarallpass}. However, the Blaschke product based
approach is only suited to the solution of the discrete-time
\emph{scalar} all-pass filter design problem, and a direct extension
of the same approach to the case of matrix all-pass filters is not
known.

Discrete-time all-pass filters are used for phase compensation in
various applications. A scalar all-pass filter can be used to correct
phase distortions in scalar signals. To the best of our knowledge,
this concept is yet to be extended to vector signals (MIMO systems),
wherein the input and output are complex vector signals, and the
transformation filter is a matrix valued all-pass filter
(unitary). These matrix all-pass filters are commonly encountered in
several situations, such as MIMO-OFDM systems and stereo audio
systems. In MIMO-OFDM communication systems, when symbols are precoded
with unitary matrices at the transmitter, if the transmitter possesses
some channel state information (CSI), unitary matrix precoding is
typically performed using the right singular vectors (or related
unitary matrices) that are obtained from the singular value
decomposition (SVD) of the channel matrix for every subcarrier
(frequency band)~\cite{TseVis}. Multiplying with a unitary matrix in
the frequency domain is norm preserving, and thus, it can be thought
of as a matrix all-pass filtering operation. In practice, having norm
preserving matrix filters is important in order to satisfy various
constraints, such as power in communication systems, or volume in the
case of audio signals, while only altering the matrix-valued
``phase''. Performing such a phase transformation using the
coefficients of an appropriate discrete-time matrix all-pass filter
with a standard LCCDE implementation would obviate the need for
frequency domain processing, and result in a more faithful
realization, and can significantly reduce the precoding complexity in
modern systems, such as those that use precoding for millimeter wave
wireless systems~\cite{majumder2021optimal,ni2020low}.

\begin{figure}
\begin{center}
\includegraphics[width=0.5\textwidth]{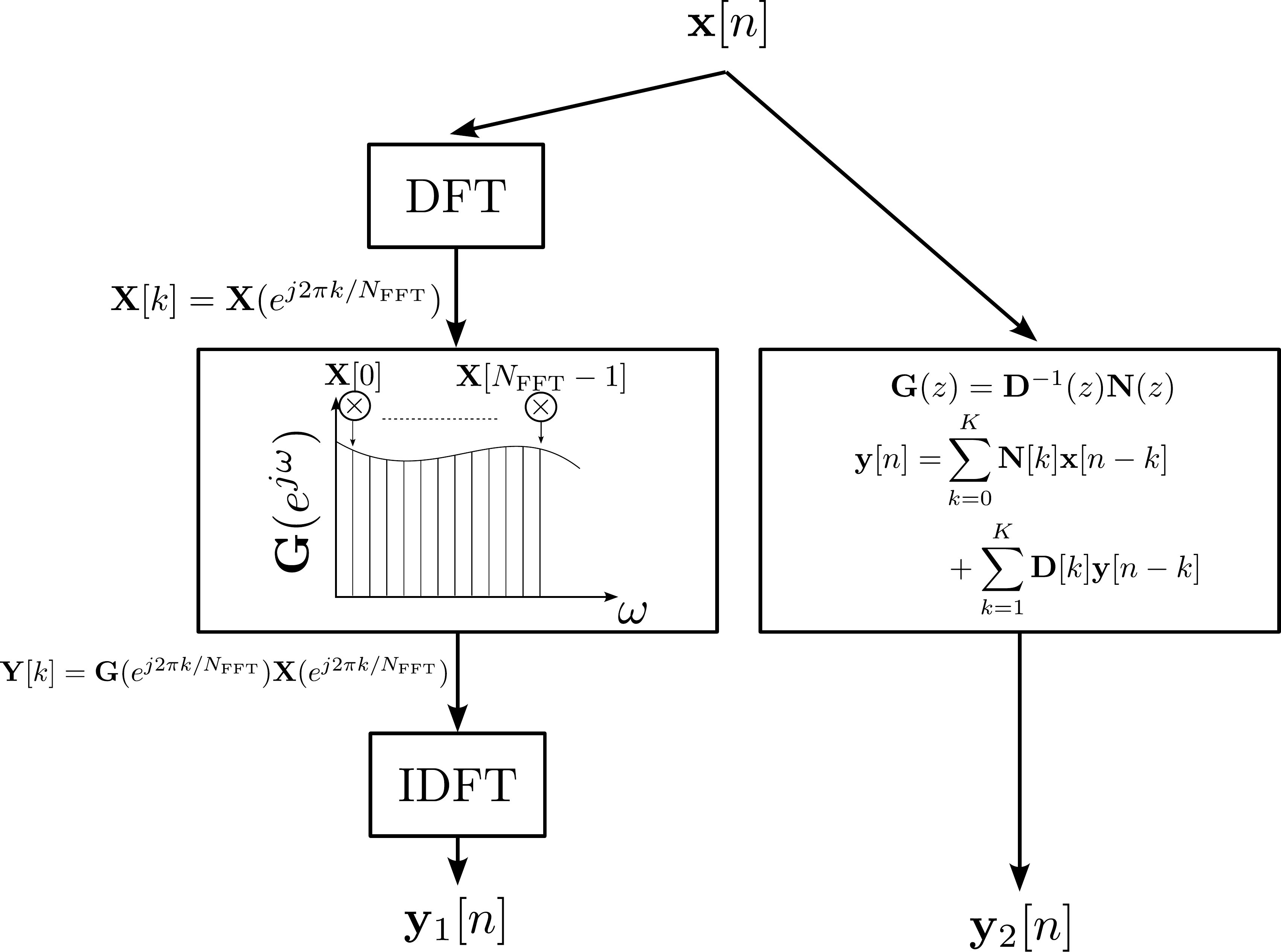}
\caption{\label{fig:block} A schematic comparing the DFT based filter
  implementations with time-domain (LCCDE) based implementations.}
\end{center}
\end{figure}
The block diagram shown in \figurename{\ref{fig:block}} depicts an unknown matrix all pass filter with efficient LCCDE representation,  with $\bx[n]$ as input and $\by[n]$ as output. Our goal is to construct a system  that mimics
the unknown matrix all-pass filter. One approach is to use DFT based techniques, wherein output $\by_1[n]$ is produced for the input $\bx[n]$. Another method is to follow a matrix filter design technique and construct a rational matrix all pass filter with LCCDE representation, and this yields output $\by_2[n]$ for the input $\bx[n]$. In this situation only the latter method is accurate.

It is evident from the block diagram \figurename{\ref{fig:block}}
that when the matrix all-pass filter has an efficient LCCDE realization, the filter realized using frequency domain techniques may be inaccurate, and may also result in less efficient realizations. An advantage of
using the SNIP based LCCDE filter design is that the computational
complexity can be reduced with a rational (realizable) all-pass filter method that has a compact representation, being more amenable to time-domain LCCDE implementations (details of the same are discussed further in Section~\ref{sec:simulations}).

In the context of control theory, realizing LCCDE matrix filters that
satisfy frequency constraints (such as bounded $\bH_\infty$ norm and prescribed matrix values at a given collection of finitely many complex numbers) is
generally accomplished using the SNIP technique~\cite{SNIP}, although this approach was primarily designed for the systems, that have non-unitary frequency
domain transfer functions. The requirement of matrix all-pass filter
challenges us to extend the SNIP to the \emph{boundary} case, wherein
the solution is unitary valued for all frequencies ($\omega$), and
thus, SNIP is not directly applicable in our case. We, thus, present a modified SNIP that accommodates the boundary case to address our
needs.

In the subsequent sections, we first discuss the SNIP approach for  continuous-time filter interpolation along with its limitations in the context of the unitary matrix valued (all-pass) constraint (boundary SNIP interpolation) due to the ill-defined Pick matrix. We then present our modified SNIP that addresses this issue and prove that we obtain a filter that satisfies the required conditions.
\section{Subspace Nevanlinna Pick Interpolation}
\label{sec:SNIP}
We now briefly outline the SNIP approach to perform filter
interpolation that is commonly used in the context of continuous-time control
systems to obtain filters that satisfy frequency domain constraints~\cite{SNIP}. We
restrict our consideration to square matrices, since our eventual
focus would be on square unitary matrix valued frequency responses.

A \emph{contractive subspace} ($\mathscr{V}_i \subset \mathbb{C}^{2m}$) is a subspace that satisfies the following property:
\begin{align*}
\resizebox{0.95\columnwidth}{!}{$\left\{ \bv = \left(\begin{array}{l}
\bv_1 \\
\bv_2
\end{array}\right) \in \mathscr{V}_i  \text{, with } \bv_1, \bv_2 \in \mathbb{C}^m  \text{ and } \bv \ne {\mathbf{0}} \right\} \Rightarrow \left\{||\bv_1||_2 > ||\bv_2||_2 \right\}.$}
\end{align*}
We consider $N$ distinct points $\lambda_i$ in the open right-half
complex plane, given together with $N$ subspaces $\mathscr{V}_i
\subset \mathbb{C}^{2m}$, where $1\leq i \leq N$.

The statement of the classical
SNIP is as follows: Given
the $N$ pairs ($\lambda_i$, $\mathscr{V}_i$), find an ${m\times m}$
polynomial matrix $\mathbf{U}$ and a non-singular ${m\times m}$
polynomial matrix $\mathbf{Y}$ such that
\begin{itemize}
    \item $\mathbf{U}$  and  $\mathbf{Y}$ are left co-prime;
    \item ($\mathbf{U}(\lambda_i) \quad \mathbf{-Y}(\lambda_i) )\bv = 0 \quad \forall \quad \bv \in \mathscr{V}_i , 1 \leq i \leq N$;
    \item and $||\mathbf{Y}^{-1}\mathbf{U}||_{\bH_\infty} < 1$.
\end{itemize}
We assume that a full column rank matrix $\mathbf{V}_i \in
\mathbb{C}^{2m \times m}$ is given such that $\text{Im}(\mathbf{V}_i)
=\mathscr{V}_i$ for every $i \in \{1, 2, \ldots N\}$. The Pick matrix
$\mathbf{P}$ for the given data, for all $j, k \in \{1, 2, \ldots N\}$, is defined as
\begin{align}
\label{eqn:SNIPpick}
\begin{split}
   &\mathbf{P}[(k-1)m+1:km , (j-1)m+1:jm] := \left( \frac{ {\mathbf{V}_k }^* \mathbf{J} {\mathbf{V}_j }    }{{\lambda_k}^* + {\lambda_j} }\right)  \\
&\text{ where } \mathbf{J} :=
\begin{bmatrix}
\mathbf{I}_{m} & \mathbf{0}_{m\times m}\\
\mathbf{0}_{m\times m} & -\mathbf{I}_{m}
\end{bmatrix}.
\end{split}
\end{align}
Here, $\bJ$ is called a \emph{signature matrix}, and $\bP[a:b,c:d]$ is a submatrix of $\bP$ obtained from contiguous
rows, ranging from the $a$th row to the $b$th row, and contiguous
columns ranging from the $c$th column to the $d$th column.

\begin{prop}
There exists a solution to the Subspace Nevanlinna Interpolation Problem (SNIP) discussed above if and only if the Hermitian matrix $\bP$ is positive definite.
\end{prop}

This result is established in Theorem 4.1 of~\cite{SNIP}.
One key disadvantage with the SNIP is that the \emph{contractive
 subspace} structure prevents it from being applicable to the case of
matrix all-pass filter design (as detailed in the following sections).
In Section~\ref{sec:DTallpass}, we adapt the SNIP framework to design
matrix all-pass filters.

\section{Modified SNIP for DT matrix all-pass filters}
\label{sec:DTallpass}
In this section, we present an adaptation of SNIP to the discrete-time
all-pass filter design case, wherein the filter constraints are
presented for a sampled system. We remark that this case is directly
applicable to several signal processing applications, such as stereo
audio processing, MIMO-OFDM precoding etc.

\emph{Note:} We refer to square matrices $\bX$ and
$\bY$ of size ${k\times k}$ as \emph{unitarily similar} if there exists an
unitary ${k\times k}$ matrix $\bT$ such that $\bY$ =  $\bT^*\bX\bT$.

\subsection{{Problem Statement}}
\label{sec:DTproblem}
The discrete-time matrix all-pass filter design problem can be stated
as follows: given data set $\mathbb{D}=\{ (\omega_{i}, \mathbf{A}_{i},
\mathbf{\Gamma}_{i} | i = 1,...,n \}$, where $\omega_{i} \in (-\pi,\pi], \mathbf{A}_{i}$
$\in \mathbb{C}^{m\times m}$ is unitary, and $\mathbf{\Gamma}_{i} \in \mathbb{C}^{m\times m}$ is a positive
definite Hermitian matrix ($\mathbf{\Gamma}_{i} =
\mathbf{\Gamma}^{*}_{i}$), we wish to obtain a rational transfer function matrix
$\mathbf{G}(z) \in \mathbb{C}(z)^{m\times m}$ that satisfies the following conditions
\begin{subequations}
\label{eqn:discrete}
\begin{align}
\begin{split}
&\label{eqn:pb1}\mathbf{G}(e^{j\omega_{i}}) = \mathbf{A}_{i}   \quad  \forall \quad i = 1,...,n,
\end{split}\\
\begin{split}
&\label{eqn:pb2}\mathbf{G}^{*}(e^{j\omega})\mathbf{G}(e^{j\omega}) = \mathbf{I}_{m}   \quad \forall \quad  \omega \in (-\pi,\pi],
\end{split}\\
\begin{split}
&\label{eqn:pb3} \text{ if }  \mathbf{F}_{i} = j\mathbf{G}^{*}(e^{j\omega_{i}})
                            \frac{d\mathbf{G}(e^{j\omega_{i}})}{d\omega}
                            \text{ then }\\
&\;\;\mathbf{F}_{i}  \text{ and }  \mathbf{\Gamma}_{i}  \text{ are unitarily similar matrices} \quad
       \forall \quad i = 1,...,n
\end{split}
\end{align}
\end{subequations}

The condition described by equation \eqref{eqn:pb2} provides us an
all-pass filter transfer function, that also matches
the desired (unitary) frequency responses $\mathbf{A}_{i}$ at the
frequencies $\omega_{i}$ as constrained by equation
\eqref{eqn:pb1}. Specifying \emph{only} these two constraints leads to
infinite possible all-pass filters (as in the scalar case~\cite{scalarallpass}). Therefore, to restrict the set of
possible solutions, we further specify $\mathbf{\Gamma}_{i}$ that are
positive definite matrices and constrain the matrices $\mathbf{F}_{i}$
at $\omega_{i}$ (equation~\eqref{eqn:pb3}). The matrices $\mathbf{F}_{i}$ correspond to the
derivative of the matrix-valued phase and thus, are referred to as \emph{group
delay matrices}~\cite{modes}. The group delay matrices $\{\mathbf{F}_{i}\}$ are
Hermitian $\forall i \in \{1,2,..n\}$. This is a direct consequence of the
unitary constraint imposed in \eqref{eqn:pb2}, and can be verified by
differentiating \eqref{eqn:pb2} with respect to $\omega$:
\begin{align*}
    \mathbf{G}^{*}(e^{j\omega})
  \frac{d\mathbf{G}(e^{j\omega})}{d\omega} +
  \frac{d\mathbf{G}^{*}(e^{j\omega})}{d\omega}\mathbf{G}(e^{j\omega})
  = \mathbf{0}_{m\times m}\\
\Rightarrow    j\mathbf{G}^{*}(e^{j\omega_{i}})\frac{d\mathbf{G}(e^{j\omega_{i}})}{d\omega} =  -j\frac{d\mathbf{G}^{*}(e^{j\omega_{i}})}{d\omega}\mathbf{G}(e^{j\omega_{i}}).
\end{align*}
It is evident from the above problem statement that the SNIP~\cite{SNIP} is very similar to
our problem, but the SNIP does not require any slope or group delay
constraints (like \eqref{eqn:pb3} in our problem statement). Moreover, in the SNIP,
the norm of the output vector is strictly smaller than the norm of the
input vector, whereas in our case, they are equal.

Thus, our current all-pass filter design problem can be considered as
an extension of the SNIP problem to the \textit{boundary} case, wherein
the norm of the output vector is exactly equal to the norm of the
input vector. This requirement prevents the SNIP from being directly
applicable to our problem, and motivates the formulation of a
modified version that can be used for boundary problems as well.

\subsection{{Formulating the modified Pick matrix}}
\label{sec:modifpickmat}
To enable a step-wise solution to the matrix all-pass interpolation
problem, we suitably modify the data set $\mathbb{D}$  to facilitate an inductive solution in subsequent steps. The data set $\mathbb{D}$ is altered by replacing $\mathbf{A}_{i}$ $\in \mathbb{C}^{m\times m}$ with $\mathbf{B}_{i}$ $\in \mathbb{C}^{2m\times m}$ that satisfy
\begin{align*}
    \text{span}_{\mathbb{C}}\mathbf{B}_{i} = \text{span}_{\mathbb{C}}
\begin{bmatrix}
\mathbf{I}_m & \mathbf{A}_{i}^{T}
\end{bmatrix}^{T}
\quad  \forall \quad i = 1,...,n.
\end{align*}
Now, in order to retain the  unitary nature of $\mathbf{A}_{i}$, we define \emph{neutral}
 $\mathbf{B}_{i}$ that satisfy
 \begin{align}
 \label{eqn:signature}
    \mathbf{B}_{i}^{*}\mathbf{J}\mathbf{B}_{i} = \mathbf{0}_{m\times m}  \qquad \text{where} \quad \mathbf{J} =
\begin{bmatrix}
\mathbf{I}_{m} & \mathbf{0}_{m\times m}\\
\mathbf{0}_{m\times m} & -\mathbf{I}_{m}
\end{bmatrix}.
\end{align}
The modified form of our data set is given by:\\ $\mathbb{D} = { \{ (\omega_{i}, \mathbf{B}_{i}, \mathbf{\Gamma}_{i}  | i = 1,...,n \}}$, where $\omega_{i} \in (-\pi,\pi]$, $\mathbf{B}_{i} $ $\in \mathbb{C}^{2m\times m}$ is neutral, $\mathbf{\Gamma}_{i} $ = $\mathbf{\Gamma}^{*}_{i} $  $\in \mathbb{C}^{m\times m}$ is positive definite.

We now set out to obtain the rational transfer function matrix $\mathbf{G}(z)$ $\in \mathbb{C}(z)^{m\times m}$  that satisfies the constraint equation \eqref{eqn:discrete}, with $\mathbf{A}_{i}$ replaced with $\mathbf{B}_{i,2}$ $\mathbf{B}_{i,1}^{-1}$, where,
\begin{align*}
\resizebox{0.95\columnwidth}{!}{$
    \mathbf{B}_{i} =
\begin{bmatrix}
\mathbf{B}_{i,1}^{T} & \mathbf{B}_{i,2}^{T}
\end{bmatrix}^{T}
\text{ and } \mathbf{B}_{i,1} , \mathbf{B}_{i,2}  \in \mathbb{C}^{m\times m}
\quad  \forall i \in \{1,2,\ldots n\}$}.
\end{align*}
We cannot use the the classical SNIP solution construction
from~\cite{SNIP} directly by substituting $\mathbf{B}_i$ and
$e^{j\omega_i}$ in place of $\mathbf{V}_i$ and $\lambda_i$
respectively in the definition of the Pick matrix block
\eqref{eqn:SNIPpick} and matrix $\bJ$ defined in
\eqref{eqn:signature}. This is because doing so results in a
$\frac{0}{0}$ form along the diagonal blocks of the Pick matrix due to
unitary nature of matrix $\mathbf{A}_i$:
\begin{align*}
\resizebox{0.95\columnwidth}{!}{$    \mathbf{P}[(i-1)m+1:im , (i-1)m+1:im] =  \frac{{\mathbf{B}_{i}}^{*}\mathbf{J}\mathbf{B}_{i} }{1 - e^{j(\omega_{i}-\omega_{i})}}$}
    \\
  \resizebox{0.35\columnwidth}{!}{$    = \frac{\mathbf{I}_m -
  {\mathbf{A}_{i}}^{*}\mathbf{A}_{i}}{1 - 1} =
  \frac{\mathbf{0}_{m\times m}}{0}$}.
\end{align*}
Therefore, we modify the Pick matrix to replace the the  ill-defined  matrices forming the diagonal blocks, and establish a set of \emph{necessary and sufficient conditions} on the modified Pick matrix for solvability of the DT matrix all-pass filter design problem.

The modified ${nm\times nm}$ Pick matrix $\mathbf{P}$ is defined as follows:
\begin{equation}
\label{eqn:pickmat}
\resizebox{0.91\columnwidth}{!}{$
    \mathbf{P}[(i-1)m+1:im , (k-1)m+1:km] :=
    \begin{cases}
    \mathbf{\Gamma}_{i}\quad,& \text{if } i = k\\
    \frac{\mathbf{B}_{i}^{*}\mathbf{J}\mathbf{B}_{k}}{1 - e^{j(\omega_{i}-\omega_{k})}},              & \text{otherwise}
    \end{cases}
    $}
\end{equation}
for all $i,k \in \{1,..,n \}$. We now state a theorem that guarantees a solution for our problem using this Pick matrix.
\begin{theorem}
  Given the data set $\mathbb{D} = \{ (\omega_{i}, \mathbf{A}_{i},
  \mathbf{\Gamma}_{i}) | i = 1,...,n \}$  with $\omega_{i}$ distinct, a rational transfer function matrix  $\mathbf{G}(z)$ $\in \mathbb{C}(z)^{m\times m}$  that satisfies \eqref{eqn:pb1}, \eqref{eqn:pb2} and \eqref{eqn:pb3} exists \emph{if and only if} the Pick matrix defined in~\eqref{eqn:pickmat} is positive definite.
\label{thm:thm2}
\end{theorem}
\begin{proof}
 Explicit construction of a solution is specified in Section~\ref{sec:interpwithgd} to prove the \emph{if} part.
 To prove the \emph{only if} part, we show that, for a given data set
 $\mathbb{D} = \{ (\omega_{i}, \mathbf{A}_{i}, \mathbf{\Gamma}_{i}) |
 i = 1,...,n \}$ if we have  $\mathbf{G}(z)$ $\in
 \mathbb{C}(z)^{m\times m}$  that satisfies \eqref{eqn:pb1},
 \eqref{eqn:pb2} and \eqref{eqn:pb3} then
 the Pick matrix for $\mathbb{D}$ defined in~\eqref{eqn:pickmat} is positive definite.

Since our interpolation problem
concerns itself with the discrete-time case, both the input
($\mathbf{x}[n]$) and output ($\mathbf{y}[n]$) to the matrix all-pass
filter are discrete-time vector signals. For simplicity, we represent
${\mathbf{x}}[n]$ as ${\mathbf{x}}$ and $ {\mathbf{y}[n]}$ as
$\mathbf{y}$, suppressing the time index. In the language of the
theory of dissipative systems, the all-pass transfer function
represents a {\em lossless} DT dynamical system with the \emph{supply
  rate} defined as $\bx^*\bx-\by^*\by$. From the fundamental theorem
of dissipative systems, it follows that lossless dissipative systems
satisfy the following equation
~\cite{dissipative,KANEKO200031,StorageISQuadratic}. That is,
 \begin{equation}
 \label{eqn:lossless}
 \resizebox{0.9\columnwidth}{!}{$
    \text{\emph{supply rate}}[n] = \text{\emph{storage function}}[n+1] - \text{\emph{storage function}}[n] $}
\end{equation}
where \{\emph{supply rate}[$n$]\} is the {supply rate} value computed at time index $n$ and similarly \{\emph{storage function}[$n$]\} is the storage function value computed at time index $n$.

Therefore, since the \emph{storage function} is a positive definite function of state variables~\cite{StorageISQuadratic}, using the concepts of discrete-time \emph{lossless} dissipative dynamical systems explained in~\cite{KANEKO200031}, we observe that
\begin{equation}
\label{eqn:disctloss}
    \sum_{n=-\infty}^{0} (\bx^*\bx - \by^*\by)  \geq 0
\end{equation}

Now, consider $\mathbf{x} = \mathbf{v}e^{({\epsilon_{i} +
    j\omega_{i}})n}$, where $\mathbf{v}$ is an arbitrary vector in
$\mathbb{C}^m$, $\epsilon_{i} $ is an arbitrary small positive real number, and
$\omega_{i}$ is the $i$-th interpolation frequency. The input to the system is $\mathbf{x}$, and hence, the output is $\mathbf{y} = \mathbf{G}\mathbf{x}$. Thus, we get $\mathbf{y} = \mathbf{G}(e^{\epsilon_{i} + j\omega_{i}})\bv e^{({\epsilon_{i} + j\omega_{i}})n}$. Let $\mathbf{G}_i:= \mathbf{G}(e^{\epsilon_{i} + j\omega_{i}}) $.
Now, substituting $\mathbf{x, y}$ in~\eqref{eqn:disctloss}, we get
\begin{align*}
    \sum_{n=-\infty}^{0} (\bx^*\bx - \by^*\by) &= \mathbf{v}^* (\bI_m - \bG_i^*\bG_i)\bv \sum_{n=-\infty}^{0} (e^{2{\epsilon_{i} }n}) \\
     &=  \frac{\mathbf{v}^* (\bI_m - \bG_i^*\bG_i)\bv}{1-e^{-2{\epsilon_{i}}}} \geq {0}.
\end{align*}
Taking limit $\epsilon_{i} \xrightarrow{}0^+$ in the above equation, we get
$\lim_{\epsilon_{i} \to 0^+}  \frac{\mathbf{v}^* (\bI_m - \bG_i^*\bG_i)\bv}{1-e^{-2{\epsilon_{i} }}} \geq \mathbf{0}$.

Now, applying L'Hospital's rule yields
\begin{align*}
  \lim_{\epsilon_{i} \to 0^+} & \frac{{\bv^* (\bI_m - \bG_i^*\bG_i)\bv}}{1-(e^{-2{\epsilon_{i}}})}   =  \lim_{\epsilon_{i} \to 0^+}  \frac{\frac{d}{d\epsilon_{i}}[\bv^* (\bI_m - \bG_i^*\bG_i)\bv]}{2}\\
    &= \lim_{\epsilon_{i} \to 0^+}  -\frac{\mathbf{v}^*\frac{d}{d\epsilon_{i}}[ \bG_i^*\bG_i)]\bv}{2}\\
    &=\lim_{\epsilon_{i} \to 0^+}  -\frac{\mathbf{v}^*[\frac{d}{d\epsilon_{i}} (\bG_i^*) \bG_i + \bG_i^*\frac{d}{d\epsilon_{i}}(\bG_i)  ]\bv}{2}.
\end{align*}
We know that, the following equalities hold,
\begin{align*}
\frac{d}{d\epsilon_{i}}\mathbf{G}_i
 = -j\frac{d}{d\omega_{i}}\mathbf{G}_i
 \text{ and }
\frac{d}{d\epsilon_{i}}(\mathbf{G}^*_i
) = j\frac{d}{d\omega_{i}}  (\mathbf{G}^*_i
).
\qquad \quad
\end{align*}
So,
\begin{align*}
\lim_{\epsilon_{i} \to 0^+}&
  \frac{\mathbf{v}^*[\frac{d}{d\epsilon_{i}} (\bG_i^*) \bG_i +
  \bG_i^*\frac{d}{d\epsilon_{i}}\bG_i  ]\bv}{2} \\ &=  j\mathbf{v}^*[ \mathbf{G}^{*}(e^{j\omega_{i}}) \frac{d\mathbf{G}(e^{j\omega_{i}})}{d\omega}  ]\bv= \mathbf{v}^*\mathbf{F}_{i}\bv \geq 0.
\end{align*}
Since $\bF_i$ and $\mathbf{\Gamma}_i$ are Hermitian and unitarily similar matrices, we obtain
\begin{align}
\label{eqn:diagPickPSD}
\mathbf{v}^*\mathbf{\Gamma}_{i}\bv \geq {0},\quad \forall \mathbf{ v} \in \mathbb{C}^{m} \text{ and } \forall i \in \{1,2,\ldots n\}.
\end{align}
It is evident that equality holds in \eqref{eqn:diagPickPSD} only if
$\mathbf{v} = \mathbf{0}_{m \times 1}$, because if $\mathbf{v}$ is non-zero, then the
corresponding $\mathbf{y}$ cannot be zero, because $\bG(e^{j\omega})$
is unitary for all $\omega \in (-\pi,\pi]$ and $\mathbf{y} = \mathbf{G}(e^{\epsilon_{i} + j\omega_{i}})\bv e^{({\epsilon_{i} + j\omega_{i}})n}$.

Next, consider  $\mathbf{x} = \sum_{i=1}^{n} \mathbf{v}_i
e^{({\epsilon_{i} + j\omega_{i}})n}$, where  $\mathbf{v}_{i}$ are
arbitrary vectors in $ \mathbb{C}^m $. Correspondingly, $\mathbf{y} = \sum_{i=1}^{n} \mathbf{G}_{i}\mathbf{v}_i e^{({\epsilon_{i} + j\omega_{i}})n}$. By a calculation similar to the one above, we get
\begin{align}
\label{eqn:DTPick}
    \resizebox{0.95\columnwidth}{!}{$
    \sum_{n=-\infty}^{0} ( \bx^*\bx - \by^*\by) = \begin{bmatrix}
    { \mathbf{v}_1^* }& { \mathbf{v}_2^*}&{...}&{\mathbf{v}_n^*}
    \end{bmatrix}
    {{\mathbf{T}}\begin{bmatrix} {
    \mathbf{v}_1^T }& {\mathbf{v}_2^T}&{...}&{\mathbf{v}_n^T}
    \end{bmatrix}}^T $}
\end{align}
where, $\mathbf{T}$ is a $nm\times nm$ matrix.
The $i$th $m\times m$ diagonal block of $\mathbf{T}$ is $[\mathbf{T}]_{ii} = \mathbf{\Gamma}_i$ as derived in \eqref{eqn:diagPickPSD}. The off-diagonal ($i,k$)th $m\times m$ block  entries of $\mathbf{T}$ are specified as follows:
\begin{align}
\label{eqn:Pickoffdiag}
    \begin{split}
    \mathbf{[T]}_{ik} = \lim_{\epsilon_{i},\epsilon_{k} \to 0^+}  &\sum_{n=-\infty}^{0}{(\bI_m - \bG^*_{i}\bG_{k})}{e^{(\epsilon_{k}+\epsilon_{i} + j(\omega_{k}-\omega_{i}))n}}\\
    &= \frac{(\bI_m - \bG^*_{i}\bG_{k})}{1-e^{j(\omega_{i}-\omega_{k})}}
    = \frac{{(\bI_m - \bA^*_{i}\bA_{k})}}{1-e^{j(\omega_{i}-\omega_{k})}}.
    \end{split}
\end{align}
Thus, we see that $\mathbf{T}$ is, in fact, the Pick matrix $\bP$ that we have defined in \eqref{eqn:pickmat}. In addition, from \eqref{eqn:disctloss} and \eqref{eqn:DTPick}, we have
\begin{align*}
    \begin{bmatrix}
    \mathbf{v}_1^* & \mathbf{v}_2^*&...&\mathbf{v}_n^*
    \end{bmatrix}
    {\mathbf{P}\begin{bmatrix}
    \mathbf{v}_1^T & \mathbf{v}_2^T&...&\mathbf{v}_n^T
    \end{bmatrix}}^T  \geq \mathbf{0}_{m\times m} \quad \forall \bv_i 
\end{align*}
where equality holds if and only if all $\mathbf{v}_{i}$s are
$\mathbf{0}_{m\times 1}$ vectors. Thus, the Pick matrix $\mathbf{P}$ defined for
data set $\mathbb{D}$ is positive definite.
\end{proof}
\subsection{Induction based solution construction}
As mentioned earlier in Theorem~\ref{thm:thm2}, if the Pick matrix
(defined in \eqref{eqn:pickmat}) for the given data set $\mathbb{D}$
is positive definite, then a matrix all-pass filter that satisfies the
constraints \eqref{eqn:pb1}, \eqref{eqn:pb2} and \eqref{eqn:pb3} can
be constructed. (Proof of the \emph{if} part of the Theorem~\ref{thm:thm2}).

We construct a solution to the problem discussed above using induction on the number of data points $n$ as follows. First, we define
\label{sec:interpwithgd}
\begin{align*}
\mathbf{B}_i = {\begin{bmatrix}
\mathbf{I}_m & \mathbf{A}_i^T
\end{bmatrix}}^T ,\qquad i = 1,2,..,n.
\end{align*}
\textbf{Base Step:} $n = 1, \mathbb{D} = (\omega_{1},
\mathbf{B}_{1}, \mathbf{\Gamma}_{1}) $.\\
Let,
${z} = e^{j\omega}$ and ${z}_1 = e^{j\omega_1}$. We first construct the following matrix polynomials. Define,
\begin{align}
\begin{split}
\mathbf{N}(z) &:= ({z}-{z}_1)\mathbf{I}_m  + \frac{\mathbf{A}_1 \mathbf{\Gamma}^{-1}_1 (\mathbf{I}_m -  \mathbf{A}^{*}_1){z}_1}{(1+{z}_1)}(1+{z})\\
\mathbf{D}(z) &:= ({z}-{z}_1)\mathbf{I}_m  + \frac{ \mathbf{\Gamma}^{-1}_1 (\mathbf{I}_m -  \mathbf{A}^{*}_1){z}_1}{(1+{z}_1)}(1+{z}).
\end{split}
\end{align}
Since, $\mathbf{\Gamma}_1$ is positive definite, $\bN(z)$ and $\bD(z)$ are well-defined.
Correspondingly, we define $\mathbf{G}(z) := \mathbf{N}(z)
\mathbf{D}(z)^{-1}$.  We now present a brief verification that confirms
that the base step of this induction indeed satisfies the required
conditions. To this end, we note that
$\mathbf{N}({z}_1)\mathbf{D}({z}_1)^{-1} = \mathbf{A}_1$, thus satisfying  \eqref{eqn:pb1}. We know that,
$\mathbf{G}^*(e^{j\omega})\mathbf{G}(e^{j\omega}) = \mathbf{I}_{m}$
if $\mathbf{N}^{*}({e^{j\omega}})\mathbf{N}({e^{j\omega}}) =
\mathbf{D}^{*}({e^{j\omega}})\mathbf{D}({e^{j\omega}})$ for all $\omega \quad \in (-\pi,\pi]$, which can be easily
verified by performing simple polynomial multiplications, thus satisfying \eqref{eqn:pb2}. Next, we see that \eqref{eqn:pb3} can be verified
as follows:
\begin{align*}
\resizebox{\columnwidth}{!}{$
 \mathbf{G}(z) \mathbf{D}(z) := \mathbf{N}(z)
\Rightarrow
 \frac{d\mathbf{G}({z}_1)}{d\omega} \mathbf{D}({z}_1) + \mathbf{G}({z}_1) \frac{d\mathbf{D}({z}_1)}{d\omega} =\frac{d\mathbf{N}({z}_1)}{d\omega}$}
\end{align*}
\begin{align*}
\Rightarrow j\bG^{*}({z}_1)\frac{d\mathbf{G}({z}_1)}{d\omega} \mathbf{D}({z}_1) =  j\bG^{*}({z}_1)\frac{d\mathbf{N}({z}_1)}{d\omega} - j\frac{d\mathbf{D}({z}_1)}{d\omega}
\end{align*}
\begin{align*}
\Rightarrow \left(j\bG^{*}({z}_1)\frac{d\mathbf{G}({z}_1)}{d\omega}\right)\left(\mathbf{\Gamma}_1^{-1} (\mathbf{I}_m
  -  \mathbf{A}^{*}_1){z}_1 \right) =  (\mathbf{I}_m
  -  \mathbf{A}^{*}_1){z}_1.
\end{align*}
Therefore, we have $ j\mathbf{G}^{*}(e^{j\omega_{1}})
\frac{d\mathbf{G}(e^{j\omega_{1}})}{d\omega} = \mathbf{\Gamma}_{1} $, and thus, $\bG(z)$ satisfies the constraint mentioned in \eqref{eqn:pb3}. Thus, the base step is verified.

\textbf{Inductive Step: } We assume that our problem is solvable for
$n - 1$ points, and use this to prove that the problem is solvable for $n$ points.
First we suitably modify the given data set of $n-1$ points. To this end, we define the following $2m\times 2m$ matrix,
\begin{align*}
\resizebox{\columnwidth}{!}{$
    {\mathbf{H}(z)  =} \begin{bmatrix}
    {2({z}-{z}_1)\mathbf{I}_m - ({z}+1)({z}_1 + 1)\mathbf{\Gamma}^{-1}_1} & {({z}+1)({z}_1 + 1)\mathbf{\Gamma}^{-1}_1\mathbf{A}_1^{*}} \\
    {-({z}+1)({z}_1 + 1)\mathbf{A}_1\mathbf{\Gamma}^{-1}_1 }&  {2({z}-{z}_1)\mathbf{I}_m +  ({z}+1)({z}_1 + 1)\mathbf{A}_1\mathbf{\Gamma}^{-1}_1\mathbf{A}_1^{*}}
    \end{bmatrix}.$}
\end{align*}
Now, for each $i \in \{1, 2, \ldots n - 1\}$, we define a modified data set as follows:
\begin{align}
\label{eqn:modifications}
\begin{split}
 \widehat{\bB}_{i} &= \mathbf{H}(e^{j\omega_{i+1}}){\begin{bmatrix}
\mathbf{I}_m & \mathbf{A}_{i+1}^T
\end{bmatrix}}^T \\
 \widehat{\mathbf{\Gamma}}_{i} &= \mathbf{\Gamma}_{i+1} -
 \frac{(\mathbf{I} -\mathbf{A}^{*}_{i+1}\mathbf{A}_{1}){\mathbf{\Gamma}_1}^{-1}(\mathbf{I}-\mathbf{A}^{*}_{1}\mathbf{A}_{i+1})}
 {(1-e^{j(\omega_{i+1}-\omega_1)})(1-e^{j(\omega_1-\omega_{i+1})})}.
\end{split}
\end{align}
Both of the above modifications result in new data sets that also
satisfy the conditions outlined in the section~\ref{sec:modifpickmat}, since  $\widehat{\bB}_i$s are neutral and $\widehat{\mathbf{\Gamma}}_i$s are still Hermitian. Thus, the modified data set is $\widehat{\mathbb{D}} := \{(\omega_{i+1},\widehat{\bB}_i,\widehat{\mathbf{\Gamma}}_i) |\; i = 1,...,n-1\}.$

It is important to note that the modifications in \eqref{eqn:modifications} are performed so that the Pick matrix of the new data set ($\widehat{\mathbb{D}}$) is the Schur complement of the Pick matrix of the original input data set ($\mathbb{D}$) with respect to $\mathbf{\Gamma}_{1}$~\cite{Schur}. Thus, using the Schur complement property, we can argue that if the original Pick matrix for the data set $\mathbb{D}$ is positive definite, then the new Pick matrix for
data set $\widehat{\mathbb{D}}$ will also be positive definite.

Using the original induction assumption, we know that we can solve for $\widehat{\mathbf{N}}(z),\widehat{\mathbf{D}}(z) \in \mathbb{C}^{m\times m}$ such that $\widehat{\mathbf{G}}(z) = \widehat{\mathbf{N}}(z){\widehat{\mathbf{D}}(z)}^{-1}$
satisfies the matrix all-pass filter design problem for the new data
set $\widehat{\mathbb{D}}$ containing $n-1$ points. Therefore, for $z_i := e^{j\omega_i}$, the following holds
\begin{align*}
\widehat{\bB}_i^*
\begin{bmatrix}
\widehat{\bD}(z_{i+1}) \\
-\widehat{\bN}(z_{i+1})
\end{bmatrix} = 0 \quad \forall \quad i = 1,\ldots,N-1
\end{align*}
\begin{align*}
    \widehat{\bB}_i = \bH(z_{i+1})\bB_{i+1} \Rightarrow    \bB_{i+1}^* {\bH(z_{i+1})}^*
\begin{bmatrix}
\widehat{\bD}(z_{i+1}) \\
-\widehat{\bN}(z_{i+1})
\end{bmatrix} = 0.
\end{align*}
We know that $\bB_1^*\bH(z_1)^* = \mathbf{0}_{m \times 2m}$. Therefore, $\forall \quad i =
1,\ldots,N$, we have
\begin{align}
\label{eqn:verfyconstr1}
\begin{bmatrix}
{\bD}(z) \\
-{\bN}(z)
\end{bmatrix}    :=  {\bH(z)}^*
\begin{bmatrix}
\widehat{\bD}(z) \\
-\widehat{\bN}(z)
\end{bmatrix} \Rightarrow  \bN(z_i)\bD(z_i)^{-1} = \bA_i.
\end{align}
We now verify the unitary nature of $\bN(e^{j\omega})\bD(e^{j\omega})^{-1}$ as follows:
\begin{align*}
\resizebox{\columnwidth}{!}{$
\bN^*(z)\bN(z)-\bD^*(z)\bD(z) =
\left[\begin{array}{ll}
\bN^*(z) & -\bD^*(z)
\end{array}\right]\left[\begin{array}{cc}
\bI_m & \mathbf{0}_{m\times m} \\
\mathbf{0}_{m\times m} & -\bI_m
\end{array}\right]\left[\begin{array}{l}
\bN(z) \\
-\bD(z)
\end{array}\right]$}
\end{align*}
\begin{align}
\begin{split}
\quad  &= \left[\begin{array}{ll}
\widehat{\bN}^*(z) & -\widehat{\bD}^*(z)
\end{array}\right] [\bH(z)^*]^{*} \bJ [\bH(z)]^{*} \left[\begin{array}{l}
\widehat{\bN}(z) \\
-\widehat{\bD}(z)
\end{array}\right]\\
\quad    &= \frac{-4(z_1 - z)^2}{z_1z}\left[\begin{array}{ll}
\widehat{\bN}^*(z) & -\widehat{\bD}^*(z)
\end{array}\right]  \bJ  \left[\begin{array}{l}
\widehat{\bN}(z) \\
-\widehat{\bD}(z)
\end{array}\right]\\
&= \frac{-4(e^{j\omega_1} - e^{j\omega} )^2}{e^{j(\omega_1+\omega)} }\left(
\widehat{\bN}^*(e^{j\omega} )\widehat{\bN}(e^{j\omega} ) -  \widehat{\bD}^*(e^{j\omega})\widehat{\bD}(e^{j\omega} )\right) \\
&= 0.
\end{split}
\end{align}
The last equality can be inferred from the original induction assumption that $\widehat{\mathbf{N}}(e^{j\omega} ){\widehat{\mathbf{D}}(e^{j\omega} )}^{-1}$
satisfies the matrix all-pass filter constraints. Therefore, \{$\bN^*(e^{j\omega})\bN(e^{j\omega})-\bD^*(e^{j\omega})\bD(e^{j\omega}) = 0\} \quad \forall \omega  \in (-\pi,\pi]$, which implies $\mathbf{N}(e^{j\omega})\mathbf{D}(e^{j\omega})^{-1} \text{ is unitary for all } \omega \in (-\pi,\pi]$.
Thus, we expand the expression of the solution for the original data set of $n$ points as follows (from equation~\eqref{eqn:verfyconstr1}):
\begin{align*}
    \mathbf{N}(z) := [2&({z}-{z}_1)\mathbf{I} - ({z}+1)({z}_1 + 1)\mathbf{A}_1\mathbf{\Gamma}^{-1}_1\mathbf{A}_1^{*}]\widehat{\mathbf{N}}(z) \\&+ [({z}+1)({z}_1 + 1)\mathbf{A}_1\mathbf{\Gamma}^{-1}_1]\widehat{\mathbf{D}}(z)\\
    \mathbf{D}(z) :=  [2&({z}-{z}_1)\mathbf{I} + ({z}+1)({z}_1 + 1)\mathbf{\Gamma}^{-1}_1]\widehat{\mathbf{D}}(z) \\&- [({z}+1)({z}_1 + 1)\mathbf{\Gamma}^{-1}_1\mathbf{A}_1^{*}]\widehat{\mathbf{N}}(z)
\end{align*}
\begin{align*}
    \mathbf{G}(z):= \mathbf{N}(z)\mathbf{D}(z)^{-1}.
\end{align*}
This $\mathbf{G}(z)$ satisfies the conditions in \eqref{eqn:pb1},
\eqref{eqn:pb2}, and \eqref{eqn:pb3}. This completes the mathematical
induction steps.

\emph{Note:} For a more intuitive understanding, we refer to
$\mathbf{\Gamma}_i$s as `group delay matrices' in the sequel, while
noting that they are unitarily similar to the original group delay
matrices $\mathbf{F}_i$ specified in \eqref{eqn:pb3}.

The group delay matrices ($\mathbf{\Gamma}_i$s) may not always be
available at the interpolating points in the given data set. In such
situations we can obtain suitable $\mathbf{\Gamma}_i$s through an
optimization process described in
Section~\ref{sec:optimization}.

\subsection{Optimizing group delay matrices at interpolating points}
\label{sec:optimization}
Theorem~\ref{thm:thm2} indicates that if the Pick matrix defined in
\eqref{eqn:pickmat} is positive definite, then an infinite number of
solutions exist that satisfy \eqref{eqn:pb3}. These solutions are
determined by specifying $\mathbf{\Gamma}_i  \forall i \in
\{1,2,\ldots n\}$. An interpretation of $\mathbf{\Gamma}_i$ is that of
a ``group delay'' matrix, mirroring the notion of group delay for
scalar filters. In the scalar case, the group delay is roughly the
delay that a signal envelope encounters when processed by a filter,
and is often sought to be minimized for several real-time
applications. In the matrix all-pass filtering case, as
considered in discussions such
as~\cite{fan2005principal,shemirani2009principal}, the group delay
matrix is related to the ``dispersion'' among the component waveforms,
representing the relative delay among them. Thus, it is prudent to
minimize a function of the group delay matrices for optimal
performance. Therefore, we choose to minimize the trace of these
matrices, since the diagonal elements of these of positive
semidefinite group delay matrices directly relates to the delay of the
various signal components.

Extending the consideration in~\cite{scalarallpass}, the trace of the
Pick matrix whose block diagonals contain the group delay matrices
$\mathbf{\Gamma}_i$ is a convex function that can be efficiently
minimized. The constraints on the trace can be expressed as linear
matrix inequalities (LMI), and semi-definite programming (SDP) can be
used to solve it efficiently. The optimization technique is to
minimize the sum of the diagonal elements in the Pick matrix, subject to
the constraint that the modified Pick matrix defined in
Section~\ref{sec:modifpickmat} is positive definite (for the selected $\mathbf{\Gamma}_i$s).
\begin{align*}
    &\text{minimize }  \quad \text{Trace}(\bP)\\
    &\text{subject to } \quad  \bP \succ \mathbf{0},\quad \mathbf{\bP}_{ii} = \mathbf{\Gamma}_{i} \succ \mathbf{0} \text{ and}\\
    & \qquad\qquad\qquad \bP_{ik} =  \frac{\mathbf{B}_{i}^{*}\mathbf{J}\mathbf{B}_{k}}{1 - e^{j(\omega_{i}-\omega_{k})}}, \quad i \ne k \\
     &\qquad\qquad\qquad \text{ for all } i,k \in \{1,2,\ldots,n \}
\end{align*}
where, $\bP_{ik} = \mathbf{P}[(i-1)m+1:im , (k-1)m+1:km]$. The key
intuition behind the formulation of this optimization is that
minimizing the sum of the trace of the group delay matrices yields a
filter that has smaller phase variations, and thus have better phase
characteristics across frequencies without abrupt variation. Moreover,
the convex nature of the problem implies that efficient tools exist
that can solve the problem fast and with adequate numerical
accuracy. Simulations confirm that the above optimization problem
yields effective realizable matrix all-pass filters.

\section{Simulation and discussion}
\label{sec:simulations}
In this section, we use precoding in MIMO-OFDM communication systems
as a sample application of the proposed SNIP approach to obtain matrix
all-pass filters, though the same generalization is applicable to the
case of other applications, such as audio, as well. We consider a
MIMO-OFDM system that uses an $N_{\text{FFT}}$ sized IDFT on a
wireless Rayleigh fading frequency selective channel. The OFDM case is
particularly interesting, since the channel within each FFT sub-band
(referred to as \textit{subcarrier}) can be assumed to be frequency
flat (i.e., flat fading). Thus, assuming that the channel in the
$k$-th subcarrier is modeled as the matrix $\bH[k] \in
\mathbb{C}^{m\times m}$, this matrix can be decomposed using the SVD
as $\bH[k] = \bU[k]\mathbf{\Sigma}[k]\bV^*[k]$, where $\mathbf{V}[k]$
is typically used to \emph{precode} (pre-multiply) the data at the
transmitter to enable channel parallelization. In other words, the
symbol vector at every subcarrier is pre-multiplied by $\mathbf{V}[k]$,
$k = 0,1,2, \ldots N_{\text{FFT}}$. However, since the channel is
typically estimated at the receiver, the receiver must feedback
$\mathbf{V}[k]$ for all subcarriers to the transmitter. This can lead
to a significant overhead in terms of feedback requirements. In
contrast, the technique from Section~\ref{sec:DTallpass} can offer two
key advantages:
\begin{itemize}
    \item By viewing $\bV[k]$ as samples of $\bV(e^{j\omega})$, $\omega \in (-\pi,\pi]$, we
can use the technique from Section~\ref{sec:DTallpass} to obtain a rational
transfer function that represents the precoding operation.
Thus, we need to feedback only the coefficients of the transfer function.
\item A compact rational transfer function would yield a simple
system of linear constant coefficient difference equations (LCCDE)
 that can be implemented in the time domain,
leading to a more compact representation. We can thereby compute the
precoder matrix by simple matrix addition and multiplication
operations, unlike the matrix exponential and logarithm operations
that are used for geodesic interpolation. More importantly, these
alternate approaches do not result in realizable filters, thereby
making them suitable for only frequency domain processing.
\end{itemize}
To emphasize these points, consider the precoding problem in a frequency
selective MIMO channel for a wireless system, wherein the matrix
precoding function $\mathbf{V}(e^{j\omega})$ has a compact
representation in terms of $\omega$. For several such randomly
generated channels, we compare the performance of our interpolation
technique with the traditional geodesic
interpolation~\cite{GeodesicInterpolation} and Givens rotation based parameterization~\cite{Flagdist,Givens_rot}. We implemented the above
discussed discrete-time matrix all-pass filter design method on a $2
\times 2$ MIMO system, with an input data set of $6$ elements. That
is, the unitary precoding filter and the corresponding group delay
matrices are known at six frequencies in $(-\pi, \pi]$.  The data set
  of unitary matrices was generated from a wireless MIMO Rayleigh
  fading channel ($\mathbf{H}(e^{j\omega})$), and taking the singular value
  decomposition as follows:
\begin{align*}
    \bH(e^{j\omega}) = \bU(e^{j\omega}) \mathbf{\Sigma}(e^{j\omega}) \bV^*(e^{j\omega}),
\end{align*}
where $\bU(e^{j\omega}) \text{ and } \bV(e^{j\omega})$ are
unitary. The channels follow the the ITU Vehicular A power delay
profile~\cite{recommendation1997guidelines}. The unitary matrices
\{$\mathbf{V}(e^{j\omega})$\} that correspond to the right singular
vectors of the channel are evaluated at ($\omega
\in$\{$-0.99\pi,-\frac{3\pi}{5},-\frac{\pi}{5},\frac{\pi}{5},\frac{3\pi}{5},0.99\pi$\}), both for the
modified SNIP approach as well as the frequency domain approaches geodesic and Givens rotation based parameterization. At frequencies other than those in data set, we interpolate
the unitary matrices using the respective techniques. To quantify the
accuracy of interpolation method, we plot the error between the
interpolated unitary matrix $\{{\hat{\bV}(e^{j\omega})}\}$ and the
unitary matrix\{${{\bV}(e^{j\omega})}$\} that is realized from the
channel matrix \{$\mathbf{H}(e^{j\omega})$\} for all $\omega$.
We remark here that the frequency domain approach does not yield
a realizable rational matrix all-pass filter, while our proposed
approach is guaranteed to do so.

\figurename{\ref{fig:frobenius_corrected_2X2}} and
\figurename{\ref{fig:flag_corrected_2X2}} present the error as
measured both using the Frobenius norm, and Flag Distance
in~\cite{Flagdist} as performance metrics.

The Flag distance gives a measure of distance between two matrices
which are considered equivalent upon multiplication by a diagonal
unitary matrices. The SVD is not unique for a matrix $\mathbf{H}(e^{j\omega})$,
since multiplying the right and left singular vectors by diagonal
unitary matrices results in equivalent
precoders~\cite{pitaval2013codebooks}. The channel capacity can be
achieved on each subcarrier by precoding with any equivalent
right-singular matrix extracted from channel matrix ($\bH(e^{j\omega})$) using the
SVD~\cite{pitaval2013coding} (due to the Flag manifold structure).

The error values in the plots are averaged over 1000 channel
iterations.
We can observe from the simulations that the output of the
discrete-time matrix all-pass filter exactly matches the unitary
matrix in data set at the frequencies in data set (referred to as
interpolating points), thereby confirming that the realizable filter
satisfies the specified frequency domain constraints.

While the performance of our discrete-time matrix all-pass filter
design technique based on SNIP (without optimizing group delays) is
comparable to geodesic interpolation and Givens rotation based parameterization, it is evident that there is
still a gap in performance when compared to the geodesic and givens rotation based approaches . This can be attributed to the fact that there may exist
other choices of group delay matrices
\{$\mathbf{\Gamma}_i$\} that are unitarily
similar matrices (see~\eqref{eqn:pb3}) which are not considered.
We address this limitation using the optimization based approach to
filter realization described in Section~\ref{sec:optimization}.
The group delay ($\mathbf{\Gamma}_i$) matrices are obtained using
optimization, and then the matrix all-pass filter is constructed using
the technique presented in~\ref{sec:interpwithgd}.
\begin{figure}
\begin{center}
\includegraphics[width=0.5\textwidth]{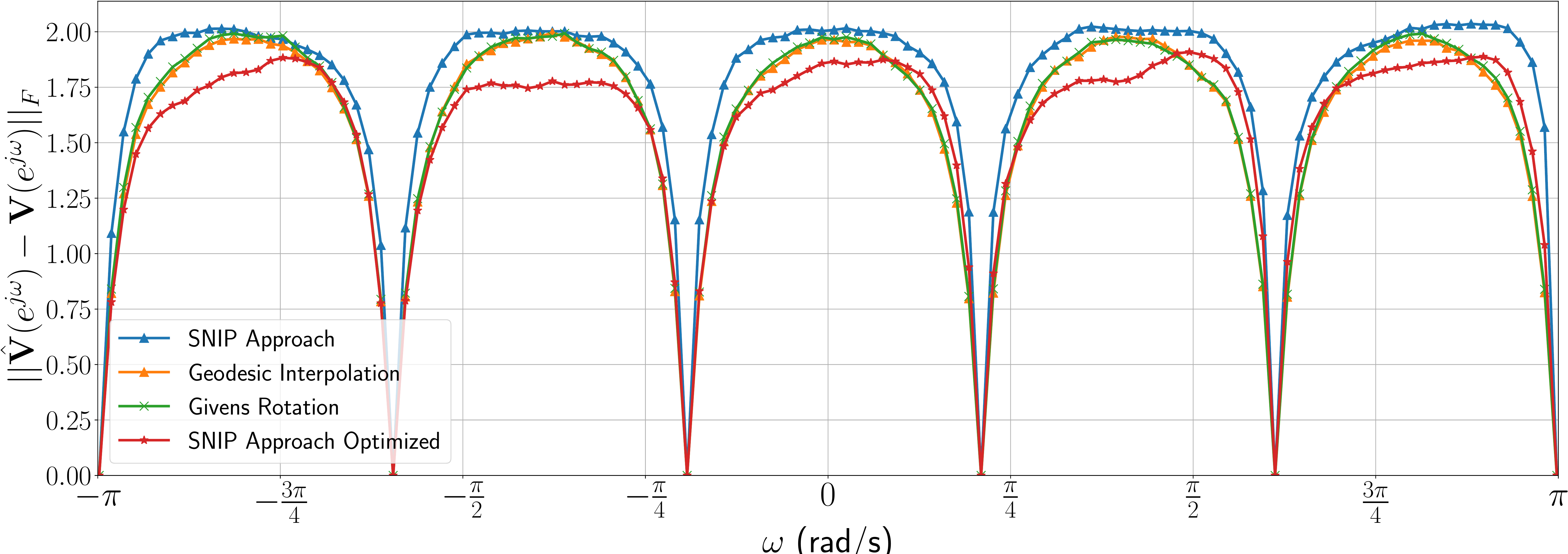}
\caption{\label{fig:frobenius_corrected_2X2}Frequency domain error
  measured for matrix all-pass filters constructed from samples for a
  $2 \times 2$ MIMO system, with Frobenius norm as error measure.}
\label{fig:frobopt}
\end{center}
\end{figure}
\begin{figure}
\begin{center}
\includegraphics[width=0.5\textwidth]{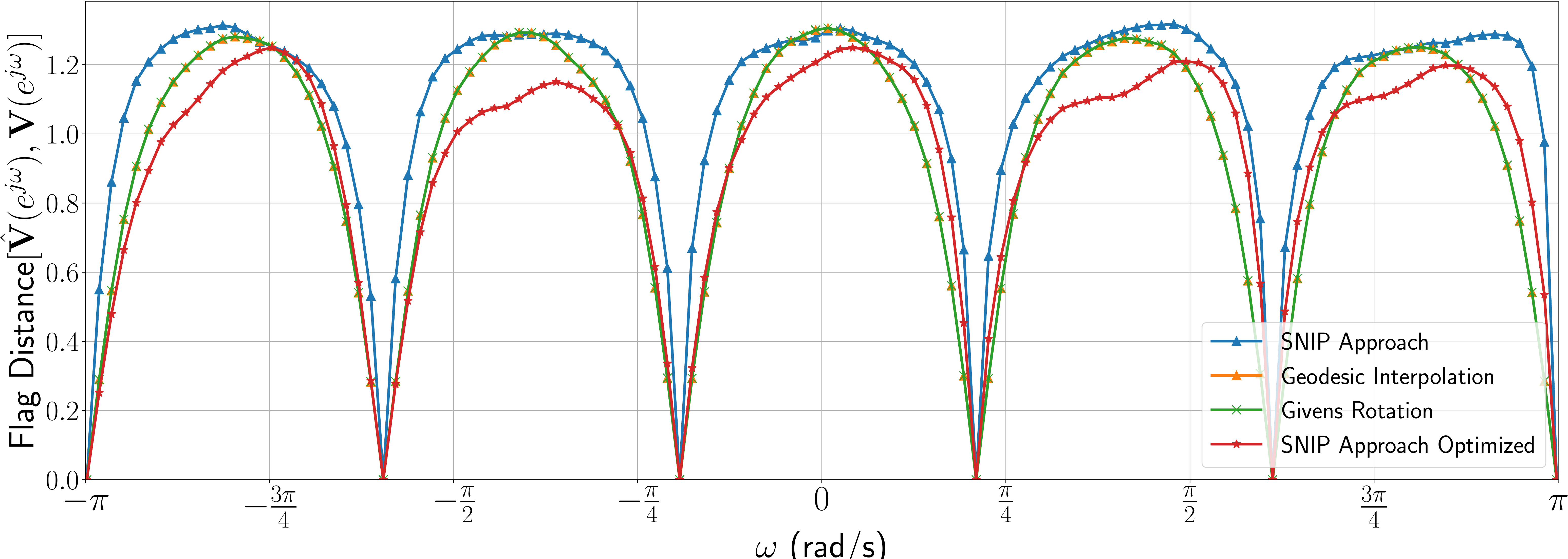}
\caption{\label{fig:flag_corrected_2X2} Frequency domain error
  measured for matrix all-pass filters constructed from samples for a
  $2 \times 2$ MIMO system, with Flag Distance as error measure.}
\label{fig:flagopt}
\end{center}
\end{figure}
\begin{figure}
\begin{center}
\includegraphics[width=0.5\textwidth]{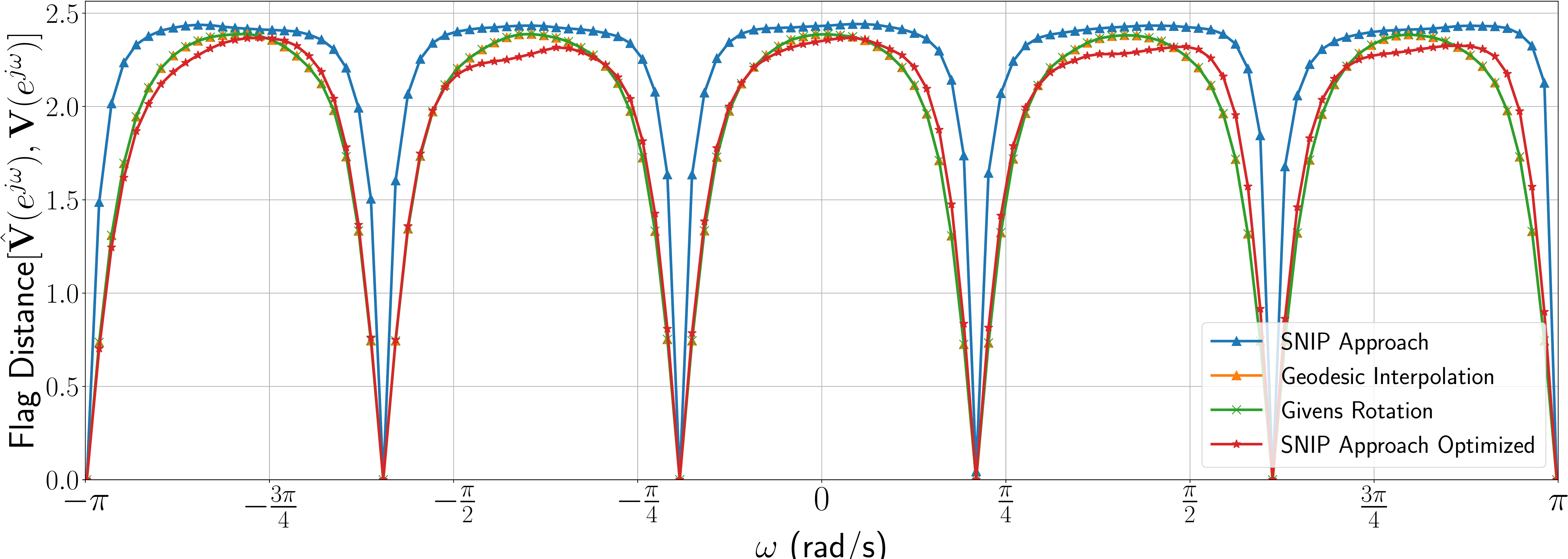}
\caption{\label{fig:flag_corrected_4X4} Frequency domain error
  measured for matrix all-pass filters constructed from samples for a
  $4 \times 4$ MIMO system, with Flag Distance as error measure.}
\end{center}
\end{figure}

\begin{figure}
\begin{center}
\includegraphics[width=0.5\textwidth]{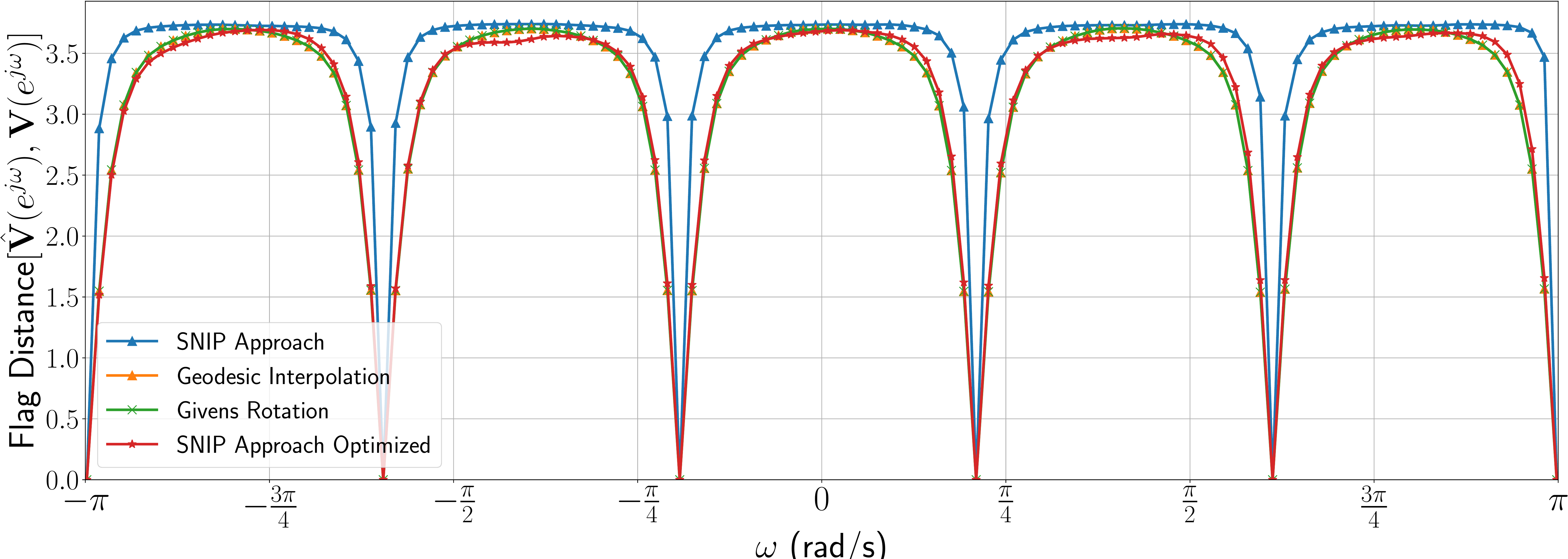}
\caption{\label{fig:flag_corrected_8X8} Frequency domain error
  measured for matrix all-pass filters constructed from samples for a
  $8 \times 8$ MIMO system, with Flag Distance as error measure.}
\end{center}
\end{figure}
It is evident from \figurename{\ref{fig:frobopt}} and
\figurename{\ref{fig:flagopt}} that optimally choosing the
$\mathbf{\Gamma}_i$ (group delay) matrices results in better
performance. Thus, by constraining the Pick matrix to be positive
definite and optimizing $\mathbf{\Gamma}_i$ matrices, we are able to
achieve significantly improved interpolation, even while obtaining
realizable (rational) matrix all-pass filters. Therefore, by using SNIP based matrix all-pass filtering technique, not only do we
have computationally efficient LCCDE realizations, but we are also
able to obtain performance comparable to a frequency-domain only
technique like geodesic interpolation (which does not yield
LCCDEs). \figurename{\ref{fig:flag_corrected_4X4}} and \figurename{\ref{fig:flag_corrected_8X8}} depicts the
corresponding performances for a $4\times 4$ system and a $8\times 8$ system respectively, this confirms that the
optimization based construction is effective for higher sizes as well
and the performance compares favourably with the geodesic technique in
spite of the realizability requirement.

In terms of complexity, using the SNIP approach would incur complexity
in evaluating $\bN(z)$, $\bD(z)$ (\emph{from their coefficients}) and finally evaluating $\bG(e^{j\omega}) =
\bN(e^{j\omega})\bD(e^{j\omega})^{-1}$, while the geodesic based
approach typically
requires the computation of
the matrix exponential for each subcarrier, thereby needing many more
computations than the SNIP approach.
To practically compare
the computational complexity of the geodesic as well as the SNIP based
approach of evaluating the frequency response at the intermediate
frequencies, we compare the time taken and memory used to construct a
precoder from the given data set. The comparisons between the the
geodesic interpolation and SNIP based matrix all-pass filter technique
are listed in the \tablename{\ref{table:time},
  \ref{table:memory}}. These values are computed on a standard Google
Colab GPU free tier, consisting of a 2 core Intel Xeon 2.2 GHz CPU with 13 GB RAM~\cite{google_colab}.

\begin{table}
\begin{center}
\tiny
\begin{tabular}{|c|c|c|c|}
\hline
{\color[HTML]{0000FF} \textbf{$m \times m$ MIMO}} & {\color[HTML]{0000FF} \textbf{SNIP Approach (in ms)}}    & {\color[HTML]{0000FF} \textbf{Givens Rotation (in ms)}}     & {\color[HTML]{0000FF} \textbf{Geodesic (in ms)}}             \\ \hline
{\color[HTML]{008000} \textbf{$m = 2$}}    & {\color[HTML]{212121} \textbf{0.5651}} & {\color[HTML]{212121} \textbf{1.7607}}  & {\color[HTML]{212121} \textbf{6.0663}} \\ \hline

{\color[HTML]{008000} \textbf{$m = 3$}}    & {\color[HTML]{212121} \textbf{0.5951}} & {\color[HTML]{212121} \textbf{3.3826}} & {\color[HTML]{212121} \textbf{7.5663}} \\ \hline

{\color[HTML]{008000} \textbf{$m = 4$}}    & {\color[HTML]{212121} \textbf{0.6371}} & {\color[HTML]{212121} \textbf{5.9818}} & {\color[HTML]{212121} \textbf{7.9641}} \\ \hline

{\color[HTML]{008000} \textbf{$m = 5$}}    & {\color[HTML]{212121} \textbf{0.6341}} & {\color[HTML]{212121} \textbf{8.8631}}  & {\color[HTML]{212121} \textbf{8.8533}} \\ \hline

{\color[HTML]{008000} \textbf{$m = 6$}}    & {\color[HTML]{212121} \textbf{0.6514}} & {\color[HTML]{212121} \textbf{12.2704}} & {\color[HTML]{212121} \textbf{9.1700}} \\ \hline

{\color[HTML]{008000} \textbf{$m = 7$}}    & {\color[HTML]{212121} \textbf{0.6518}} & {\color[HTML]{212121} \textbf{16.1756}} & {\color[HTML]{212121} \textbf{9.5353}} \\ \hline

\end{tabular}
\caption{Average time taken to construct a precoder matrix from the feedback data-set}
\label{table:time}
\end{center}
\end{table}

\begin{table}
\begin{center}
\tiny
\begin{tabular}{|c|c|c|c|}
\hline
{\color[HTML]{0000FF} \textbf{$m \times m$ MIMO}} & {\color[HTML]{0000FF} \textbf{SNIP Approach (in KB)}}  & {\color[HTML]{0000FF} \textbf{Givens Rotation (in KB)}}           & {\color[HTML]{0000FF} \textbf{Geodesic (in KB)}}                  \\ \hline
{\color[HTML]{008000} \textbf{$m = 2$}}    & {\color[HTML]{212121} \textbf{3.3347}} & {\color[HTML]{212121} \textbf{7.2606}} & {\color[HTML]{212121} \textbf{7.0726}}  \\ \hline

{\color[HTML]{008000} \textbf{$m = 3$}}    & {\color[HTML]{212121} \textbf{4.4608}}  & {\color[HTML]{212121} \textbf{8.2336}}  & {\color[HTML]{212121} \textbf{7.7833}} \\ \hline

{\color[HTML]{008000} \textbf{$m = 4$}}    & {\color[HTML]{212121} \textbf{5.7499}}  & {\color[HTML]{212121} \textbf{9.3939}}  & {\color[HTML]{212121} \textbf{8.700}}   \\ \hline

{\color[HTML]{008000} \textbf{$m = 5$}}    & {\color[HTML]{212121} \textbf{6.9731}}  & {\color[HTML]{212121} \textbf{10.3815}} & {\color[HTML]{212121} \textbf{9.9756}}  \\ \hline

{\color[HTML]{008000} \textbf{$m = 6$}}    & {\color[HTML]{212121} \textbf{8.2406}}  & {\color[HTML]{212121} \textbf{12.1696}} & {\color[HTML]{212121} \textbf{12.9178}} \\ \hline

{\color[HTML]{008000} \textbf{$m = 7$}}    &
{\color[HTML]{212121} \textbf{9.6660}}  & {\color[HTML]{212121} \textbf{13.8401}} & {\color[HTML]{212121} \textbf{16.2282}}  \\ \hline
\end{tabular}
\caption{Average peak size used to construct a precoder matrix from the feedback data-set}
\label{table:memory}
\end{center}
\end{table}
It is evident from the simulation results that, SNIP based approach
takes much less time and consumes less memory than geodesic
interpolation to construct a precoder (unitary matrix). This
translates to much lower complexity, especially for modern millimeter
wave communication systems that employ large antenna arrays and FFT sizes.

\section{Conclusion}
\label{sec:conclusion}
All-pass filtering of vector signals is typically done in the
frequency domain using the DFT. However, this could lead to
incorrect filtering, and may require more complex representation of
filters.  In this paper, we have presented a method to obtain a
realizable (rational) matrix all-pass filter by extending the SNIP to
the boundary case, leading to an LCCDE implementation.  If the matrix
valued phase response is specified only at $n$ distinct values of
$\Omega$ (or $\omega$), we can obtain an $n$ pole all-pass filter that
exactly satisfies the phase constraints at these points. In addition,
we show that, if the values of the group delay for the interpolating
$\omega$ are specified,
while often
ensuring a compact representation.  The group delay matrices at the
interpolating points in the problem statement can be optimized or
tuned to control the phase response for the remaining frequencies,
which is also done in~\cite{scalarallpass} for the scalar case.
Simulations reveal that the method proposed can significantly
outperform other optimization based approaches with much lower
complexity. Future work would focus on stability of the solution
under perturbation and approaches to minimize filter order as well as
the problem of generating a good data sets that minimize the
overfitting and undersampling for this boundary SNIP-based interpolation.

\renewcommand{\bibfont}{\footnotesize}
\bibliographystyle{IEEEtran} \bibliography{IEEEabrv,references}

\begin{thebibliography}{1}
\providecommand{\url}[1]{#1}
\csname url@samestyle\endcsname
\providecommand{\newblock}{\relax}
\providecommand{\bibinfo}[2]{#2}
\providecommand{\BIBentrySTDinterwordspacing}{\spaceskip=0pt\relax}
\providecommand{\BIBentryALTinterwordstretchfactor}{4}
\providecommand{\BIBentryALTinterwordspacing}{\spaceskip=\fontdimen2\font plus
\BIBentryALTinterwordstretchfactor\fontdimen3\font minus
  \fontdimen4\font\relax}
\providecommand{\BIBforeignlanguage}[2]{{%
\expandafter\ifx\csname l@#1\endcsname\relax
\typeout{** WARNING: IEEEtran.bst: No hyphenation pattern has been}%
\typeout{** loaded for the language `#1'. Using the pattern for}%
\typeout{** the default language instead.}%
\else
\language=\csname l@#1\endcsname
\fi
#2}}
\providecommand{\BIBdecl}{\relax}
\BIBdecl

\bibitem{bolotnikov2018boundary}
V.~Bolotnikov, ``{Boundary Interpolation by Finite Blaschke Products},'' in
  \emph{Complex Analysis and Dynamical Systems}.\hskip 1em plus 0.5em minus
  0.4em\relax Springer, 2018, pp. 39--65.

\end{thebibliography}


\begin{thebibliography}{10}
\providecommand{\url}[1]{#1}
\csname url@samestyle\endcsname
\providecommand{\newblock}{\relax}
\providecommand{\bibinfo}[2]{#2}
\providecommand{\BIBentrySTDinterwordspacing}{\spaceskip=0pt\relax}
\providecommand{\BIBentryALTinterwordstretchfactor}{4}
\providecommand{\BIBentryALTinterwordspacing}{\spaceskip=\fontdimen2\font plus
\BIBentryALTinterwordstretchfactor\fontdimen3\font minus
  \fontdimen4\font\relax}
\providecommand{\BIBforeignlanguage}[2]{{%
\expandafter\ifx\csname l@#1\endcsname\relax
\typeout{** WARNING: IEEEtran.bst: No hyphenation pattern has been}%
\typeout{** loaded for the language `#1'. Using the pattern for}%
\typeout{** the default language instead.}%
\else
\language=\csname l@#1\endcsname
\fi
#2}}
\providecommand{\BIBdecl}{\relax}
\BIBdecl

\bibitem{precodingMIMO}
D.~Love and R.~Heath, ``Multimode precoding for mimo wireless systems,''
  \emph{{IEEE} Trans. Signal Process.}, vol.~53, no.~10, pp. 3674--3687, 2005.

\bibitem{LinearPrecodMIMO}
H.~Karaa, R.~S. Adve, and A.~J. Tenenbaum, ``Linear precoding for multiuser
  mimo-ofdm systems,'' in \emph{2007 IEEE International Conference on
  Communications}, 2007, pp. 2797--2802.

\bibitem{LimitFeedbackPrecode}
D.~J. Love, R.~W. Heath, V.~K. N.~Lau, D.~Gesbert, B.~D. Rao, and M.~Andrews,
  ``An overview of limited feedback in wireless communication systems,''
  \emph{{IEEE} J. Sel. Areas Commun.}, vol.~26, no.~8, pp. 1341--1365, 2008.

\bibitem{feedback}
D.~J. Love and R.~W. Heath, ``{Limited feedback unitary precoding for spatial
  multiplexing systems},'' \emph{{IEEE} Trans. Inf. Theory}, vol.~51, no.~8,
  pp. 2967--2976, 2005.

\bibitem{TseVis}
D.~Tse and P.~Viswanath, \emph{{Fundamentals of Wireless Communication}}.\hskip
  1em plus 0.5em minus 0.4em\relax {Cambridge University Press}, 2005.

\bibitem{SNIP}
P.~Rapisarda and J.~C. Willems, ``{The Subspace Nevalinna Interpolation Problem
  and the most powerful unfalsified model},'' in \emph{36th IEEE Conference on
  Decision and Control}, vol.~3, 1997, pp. 2029--2033.

\bibitem{interp_posit_fns}
M.~S. D.C.~Youla, ``{Interpolation with positive real functions},''
  \emph{Journal of the Franklin Institute}, vol. 284, Issue 2, pp. 77--108,
  1967.

\bibitem{scalarallpass}
K.~Appaiah and D.~Pal, ``{All-pass filter design using Blaschke
  interpolation},'' \emph{{IEEE} Signal Process. Lett.}, vol.~27, pp. 226--230,
  2020.

\bibitem{bolotnikov2018boundary}
V.~Bolotnikov, ``{Boundary Interpolation by Finite Blaschke Products},'' in
  \emph{Complex Analysis and Dynamical Systems}.\hskip 1em plus 0.5em minus
  0.4em\relax Springer, 2018, pp. 39--65.

\bibitem{lou2013comparison}
H.~Lou, M.~Ghosh, P.~Xia, and R.~Olesen, ``{A comparison of implicit and
  explicit channel feedback methods for MU-MIMO WLAN systems},'' in \emph{24th
  Annual Intl. Symp. on Personal, Indoor, and Mobile Radio Commun.
  (PIMRC)}.\hskip 1em plus 0.5em minus 0.4em\relax IEEE, 2013, pp. 419--424.

\bibitem{ieee80211}
{IEEE P802.11 - Task Group}, ``{IEEE P802.11ax Part 11: Wireless LAN MAC and
  PHY specifications. Amendment 1: Enhancements for High Efficiency WLAN},''
  \emph{IEEE P802.11 - Task Group AX}, 2019.

\bibitem{pitaval2013coding}
R.-A. Pitaval, ``Coding on flag manifolds for limited feedback mimo systems,''
  2013.

\bibitem{Flagdist}
S.~Nijhawan, A.~Gupta, K.~Appaiah, R.~Vaze, and N.~Karamchandani, ``{Flag
  Manifold-Based Precoder Interpolation Techniques for MIMO-OFDM Systems},''
  \emph{{IEEE} Trans. Commun.}, vol.~69, no.~7, pp. 4347--4359, July 2021.

\bibitem{4114278}
J.~C. {Roh} and B.~D. {Rao}, ``Efficient feedback methods for mimo channels
  based on parameterization,'' \emph{{IEEE} Trans. Wireless Commun.}, vol.~6,
  no.~1, pp. 282--292, Jan 2007.

\bibitem{Givens_rot}
M.~Madan, A.~Gupta, and K.~Appaiah, ``Scalar feedback-based joint
  time-frequency precoder interpolation for mimo-ofdm systems,'' \emph{IEEE
  Wireless Communications Letters}, vol.~9, no.~9, pp. 1562--1566, 2020.

\bibitem{majumder2021optimal}
M.~Majumder, H.~Saxena, S.~Srivastava, and A.~K. Jagannatham, ``Optimal bit
  allocation-based hybrid precoder-combiner design techniques for mmwave
  mimo-ofdm systems,'' \emph{IEEE Access}, vol.~9, pp. 54\,109--54\,125, 2021.

\bibitem{ni2020low}
Z.~Ni, J.~A. Zhang, K.~Yang, F.~Gao, and J.~An, ``Low-complexity subarray-based
  rf precoding for wideband multiuser millimeter wave systems,'' \emph{{IEEE}
  Trans. Veh. Technol.}, vol.~69, no.~7, pp. 8028--8033, 2020.

\bibitem{modes}
S.~Fan and J.~M. Kahn, ``Principal modes in multimode waveguides,,''
  \emph{Optics Letters}, vol.~30, no.~2, pp. 135--137, 2005.

\bibitem{dissipative}
\BIBentryALTinterwordspacing
J.~C. Willems., ``{Dissipative Dynamical Systems Part I: General Theory},''
  \emph{Archive for Rational Mechanics and Analysis}, vol.~45, no.~5, pp.
  321--351, Jan. 1972. [Online]. Available:
  \url{https://doi.org/10.1007/BF00276493}
\BIBentrySTDinterwordspacing

\bibitem{KANEKO200031}
\BIBentryALTinterwordspacing
O.~Kaneko and T.~Fujii, ``Discrete-time average positivity and spectral
  factorization in a behavioral framework,'' \emph{Systems \& Control Letters},
  vol.~39, no.~1, pp. 31--44, 2000. [Online]. Available:
  \url{https://www.sciencedirect.com/science/article/pii/S0167691199000808}
\BIBentrySTDinterwordspacing

\bibitem{StorageISQuadratic}
\BIBentryALTinterwordspacing
H.~Trentelman and J.~Willems, ``Every storage function is a state function,''
  \emph{Systems \& Control Letters}, vol.~32, no.~5, pp. 249--259, 1997, system
  and Control Theory in the Behavioral Framework. [Online]. Available:
  \url{https://www.sciencedirect.com/science/article/pii/S0167691197000819}
\BIBentrySTDinterwordspacing

\bibitem{Schur}
J.~H. Gallier, ``{Notes on the Schur Complement},'' December 2010.

\bibitem{fan2005principal}
S.~Fan and J.~M. Kahn, ``Principal modes in multimode waveguides,''
  \emph{Optics letters}, vol.~30, no.~2, pp. 135--137, 2005.

\bibitem{shemirani2009principal}
M.~B. Shemirani, W.~Mao, R.~A. Panicker, and J.~M. Kahn, ``Principal modes in
  graded-index multimode fiber in presence of spatial-and polarization-mode
  coupling,'' \emph{J. Lightw. Technol.}, vol.~27, no.~10, pp. 1248--1261,
  2009.

\bibitem{GeodesicInterpolation}
J.~Choi, B.~Mondal, and R.~W. Heath, ``Interpolation based unitary precoding
  for spatial multiplexing mimo-ofdm with limited feedback,'' \emph{{IEEE}
  Trans. Signal Process.}, vol.~54, no.~12, pp. 4730--4740, 2006.

\bibitem{recommendation1997guidelines}
I.-R. Recommendation, ``Guidelines for evaluation of radio transmission
  technologies for imt-2000,'' \emph{Rec. ITU-R M. 1225}, 1997.

\bibitem{pitaval2013codebooks}
R.-A. Pitaval, A.~Srinivasan, and O.~Tirkkonen, ``Codebooks in flag manifolds
  for limited feedback mimo precoding,'' in \emph{SCC 2013; 9th International
  ITG Conference on Systems, Communication and Coding}.\hskip 1em plus 0.5em
  minus 0.4em\relax VDE, 2013, pp. 1--5.

\bibitem{google_colab}
\BIBentryALTinterwordspacing
G.~Inc., ``Google colaboratory.'' [Online]. Available:
  \url{https://colab.research.google.com/}
\BIBentrySTDinterwordspacing

\end{thebibliography}
\end{document}